\documentclass[sigconf, nonacm]{acmart}

\newcommand\vldbdoi{10.14778/3467861.3467872}
\newcommand\vldbpages{XXX-XXX}
\newcommand\vldbvolume{14}
\newcommand\vldbissue{10}
\newcommand\vldbyear{2021}
\newcommand\vldbauthors{\authors}
\newcommand\vldbtitle{\shorttitle} 
\newcommand\vldbavailabilityurl{https://github.com/AwesomeYifan/Data-acquisition-for-ML}
\newcommand\vldbpagestyle{empty}

\usepackage{booktabs}
\usepackage{graphicx}
\usepackage{balance}
\usepackage{hhline}
\usepackage{float}
\usepackage{multirow}
\usepackage{amsmath}
\usepackage{algorithm}
\usepackage{algorithmicx}
\usepackage[utf8]{inputenc}
\usepackage[noend]{algpseudocode}
\usepackage{subcaption}
\usepackage{bbm}
\usepackage[skip=0pt]{caption}

\begin{document}
\title{Data Acquisition for Improving Machine Learning Models}


\author{Yifan Li}
\affiliation{%
  \institution{York University}
}
\email{yifanli@eecs.yorku.ca}

\author{Xiaohui Yu}
\affiliation{%
  \institution{York University}
}
\email{xhyu@yorku.ca}

\author{Nick Koudas}
\affiliation{%
  \institution{University of Toronto}
}
\email{koudas@cs.toronto.edu}

\begin{abstract}

The vast advances in Machine Learning (ML) over the last ten years have been powered by the availability of suitably prepared data for training purposes. The future of ML-enabled enterprise hinges on data. As such, there is already a vibrant market offering data annotation services to tailor sophisticated ML models.

In this paper, inspired by the recent vision of online data markets and associated market designs, we present research on the practical problem of obtaining data in order to improve the accuracy of ML models. We consider an environment in which consumers query for data to enhance the accuracy of their models and data providers who possess data make them available for training purposes. We first formalize this interaction process laying out the suitable framework and associated parameters for data exchange. We then propose two data acquisition strategies that consider a trade-off between {\em exploration} during which we obtain data to learn about the distribution of a provider's data and {\em exploitation} during which we optimize our data inquiries utilizing the gained knowledge. In the first strategy, {\em Estimation and Allocation} (EA), we utilize queries to estimate the utilities of various predicates while learning about the distribution of the provider's data; then we proceed to the allocation stage in which we utilize those learned utility estimates to inform our data acquisition decisions. The second algorithmic proposal, named {\em Sequential Predicate Selection} (SPS), utilizes a sampling strategy to explore the distribution of the provider's data, adaptively investing more resources to parts of the data space that are statistically more promising to improve overall model accuracy.

We present a detailed experimental evaluation of our proposals utilizing a variety of ML models and associated real data sets exploring all applicable parameters of interest. Our results demonstrate the relative benefits of the proposed algorithms. Depending on the models trained and the associated learning tasks we identify trade-offs and highlight the relative benefits of each algorithm to further optimize model accuracy. 

\end{abstract}

\maketitle

\pagestyle{\vldbpagestyle}
\begingroup\small\noindent\raggedright\textbf{PVLDB Reference Format:}\\
\vldbauthors. \vldbtitle. PVLDB, \vldbvolume(\vldbissue): \vldbpages, \vldbyear.\\
\href{https://doi.org/\vldbdoi}{doi:\vldbdoi}
\endgroup
\begingroup
\renewcommand\thefootnote{}\footnote{\noindent
This work is licensed under the Creative Commons BY-NC-ND 4.0 International License. Visit \url{https://creativecommons.org/licenses/by-nc-nd/4.0/} to view a copy of this license. For any use beyond those covered by this license, obtain permission by emailing \href{mailto:info@vldb.org}{info@vldb.org}. Copyright is held by the owner/author(s). Publication rights licensed to the VLDB Endowment. \\
\raggedright Proceedings of the VLDB Endowment, Vol. \vldbvolume, No. \vldbissue\ %
ISSN 2150-8097. \\
\href{https://doi.org/\vldbdoi}{doi:\vldbdoi} \\
}\addtocounter{footnote}{-1}\endgroup

\ifdefempty{\vldbavailabilityurl}{}{
\vspace{.3cm}
\begingroup\small\noindent\raggedright\textbf{PVLDB Artifact Availability:}\\
The source code, data, and/or other artifacts have been made available at \url{\vldbavailabilityurl}.
\endgroup
}

\section{Introduction}

\textbf{Background and Motivation.} Data traditionally has been an asset in deriving projections or making decisions. The prevalence of Machine Learning (ML) models across business functions necessitates access to ample and diverse data sources for training. As an answer to the vast demand for training data, numerous businesses offer data annotation services \cite{mturk,hive,appen} providing annotated data in a myriad of business categories with varying degrees of specialization. It is evident that the need for specialized data to train ML models has created a corresponding market fulfilling the purpose.

At the same time, in recent years we have experienced the increasing prevalence of {\em online data markets} such as Dawex \cite{dawex}, World-Quant \cite{worldquant}, and Xignite \cite{xignite}, to name a few, which aim to make access to data a commodity,  for modelling or learning purposes. In these markets the main idea is to facilitate interaction between {\em data providers} (e.g., individuals or organizations that possess data in diverse domains and wish to offer them to other interested parties) and {\em data consumers} who are interested to obtain data to accomplish certain tasks, such as training new ML models or increasing the accuracy of existing ones, or conducting statistical estimation. Since such platforms aim to adopt the characteristics of a market, data exchange carries an underlying cost (e.g., monetary value). The emergence of such markets can be viewed as an initial step to the enablement of efficient trading of data.

The design of the operating principles, market mechanisms and tradings strategies (to name a few topics) of such markets constitute open research directions and involve multiple research communities. Recently, in the database community, Fernandez et al. \cite{fernandez2020data} presented their vision for a research agenda on market design in data market platforms and discussed various important research directions of broader data management interest in making the data market vision a reality.

\textbf{The Problem.} In this paper, we consider a domain $\Gamma=\{\mathcal{X},\mathcal{Y}\}$, where $\mathcal{X}$ denotes the feature space and $\mathcal{Y}$ denotes the label space (e.g., possible class labels for classification tasks, possible values of the dependent variable for regression tasks). The purpose is to train a model for a target distribution $p$ over $\Gamma$ such that the model attains high accuracy on data drawn from the same distribution.  We focus on supervised learning and assume that the data consumer (\textit{consumer} for short) already has an ML model (e.g., CNN, SVM, a regression model) built utilizing some training data from $\Gamma$, and wishes to obtain data from a data provider (\textit{provider}) offering $\mathcal{D}_{pool}$ drawn from the target distribution $p$. 
{The aim of the consumer is to maximize  the improvement in the accuracy of the model\footnote{We use the term ``accuracy" in a broad sense here, which depending on the ML model can be measured in different ways (precision for classifiers, root mean squared error for regression models, etc.), and our discussion is independent of its exact choice.}. }

{To facilitate the interaction, the provider exposes meta-data of $\mathcal{D}_{pool}$, such as the range of values in each attribute, and the set of possible labels on the records.} 
We assume a typical query interface supported by both parties, akin to the prevalent application programming interfaces (API) in existence for any online service \cite{twitter}. 
The interface supports a predicate $P$ specifying the properties of the records requested, and an integer $I$ denoting the number of records to obtain (e.g., 10 images with \texttt{label = 'dog'}, or 100 records with \texttt{2018$\le$year$\le$2020}). Such predicates impose multi-attribute conjunctive conditions in the more general case. After receiving the query request from the consumer, the provider randomly selects without replacement $I$ records satisfying $P$  from $\mathcal{D}_{pool}$ and returns these records to the consumer. 

We assume that the consumer carries a {\em budget} $B$ on the total number of records that can be requested\footnote{Such a budget can be determined based on monetary costs per record offered by the provider or by the monetary cost of each query, etc. Any mechanism to assign a value to data (e.g., price or otherwise) is completely orthogonal to our approach.}. The {number of }records  requested by the consumer each time a query is issued to the provider may vary, and is decided by the consumer, as long as the total number of records requested across all queries is within the budget $B$.
For example, if the consumer can obtain 10 images from the provider, these can be obtained by inquiring for 10 images once or  for 5 images twice.  

Suppose that the accuracy of the underlying model is evaluated on a testing data set $\mathcal{D}_{test}\subset\Gamma$. The task of the consumer is to identify a series of queries $\langle(P_1,I_1),(P_2,I_2),\cdots,(P_z,I_z)\rangle$ to obtain $B$ records, where $\sum_{i=1}^z I_i =B$ with $P_i$ and $I_i$ being respectively the predicate and the number of requested records in the $i$-th query. 
Let $\mathcal{D}_{otd}$ be the records  obtained from the provider using the identified queries  ($\vert\mathcal{D}_{otd}\vert=B$), $\mathcal{D}_{init}$ be the data the model of the consumer is trained on initially, and $\mathcal{M}'$ be the model re-trained on all the data the consumer has after data acquisition (i.e., $\mathcal{D}_{otd} \cup~ \mathcal{D}_{init}$ ). The objective of the data acquisition process is to improve as much as possible the accuracy of $\mathcal{M}'$ on $\mathcal{D}_{test}$.

\textbf{Proposed Solutions.} We develop data acquisition strategies to address this problem. In particular, we consider the trade-off between \textit{exploration} and {\em exploitation} in data acquisition.  During exploration, requested data records from the available budget are obtained to gain more knowledge regarding the distribution the provider's data, so that better predicates can be designed for subsequent queries. During exploitation, data records are obtained based on the current information the consumer possesses. With a limited budget of records to be requested, one must strike a balance between exploration and exploitation by allocating the requested records within the budget wisely such that the accuracy of the resulting model is maximized.  


We propose two methods to determine how to allocate the existing budget of records across queries (the budget allocation problem) adopting different strategies. The first solution, which we refer to as {\em estimation-and-allocation} (EA), consists of two stages: during the Estimation Stage, the consumer issues a number of queries obtaining a number of random records for each of them to explore from $\mathcal{D}_{pool}$; we subsequently estimate (without re-training the model) the expected
improvement in model accuracy  utilizing the records for each query, which we refer to as {\em predicate utility}. During the Allocation Stage, the consumer allocates the remaining record budget according to the estimated utilities. We investigate methods to quantify the estimation error and propose an adaptive method to balance between reducing the estimation error and controlling how much of the record budget is devoted to obtaining the estimation; as a result, budget is reserved to be allocated more effectively for predicates with high utilities. 
For the Allocation Stage, we propose different allocation strategies and showcase their performance under various  settings in Section~\ref{sec:exp}. 

The second solution, which we refer to as {\em sequential predicate selection} (SPS), is based on the 
observation that for a predicate $P$, the associated predicate utility decreases as we obtain more records for the predicate, due to information redundancy \cite{guo2007discriminative}. The core idea of SPS is to iteratively pose queries requesting a small number of records while balancing between (1) obtaining more records with predicates yielding higher expected utility, and (2) closely monitoring the utility decrease of each predicate as we obtain more records. We implement this design utilizing Thompson Sampling (TS) \cite{thompson1933likelihood,russo2017tutorial}, an action selection method, for its simplicity and proven performance, but the design can be implemented with other action selection methods (such as the $\epsilon$-greedy algorithm \cite{sutton2018reinforcement}) as well. 

As both EA and SPS rely on the expected predicate utility, we investigate how to best estimate it without re-training the underlying ML model. The utility of a predicate $P$ is essentially measured by the improvement in model accuracy resulting from the set of records selected by $P$,  $\mathcal{R}_P$. We propose {\em novelty}, which describes how different the distribution of $\mathcal{R}_P$ is from the distribution of the records the consumer currently possesses satisfying $P$, as the indicator of the potential accuracy gain $\mathcal{R}_P$ brings to the model. Note that our subsequent discussion applies to other utility measures as well; we experimentally compare various measures in Section \ref{exp:measure}.

We evaluate the performance of EA and SPS on both traditional ML tasks and Deep Learning tasks, including spatial regression, radar data classification, and image classification, using classical ML models as well as state-of-the-art deep models. As will be shown in Section \ref{sec:exp}, the proposed methods demonstrate solid performance across a variety of settings, outperforming alternative approaches that require frequent model re-training. We also thoroughly study the effects of  various parameters on EA and SPS and provide suggestions on their settings in real-world scenarios.

\textbf{Contributions.} Our main contributions can be summarized as follows.
\begin{itemize}
    \item We formally define and study the problem of data acquisition for improving the performance of ML models given a budget. We consider this problem in the context of a data consumer and a data provider in a data market, and it can serve as a building block for a variety of data markets.
    \item We propose an estimation-and-allocation solution, EA, which first  estimates the utility of each predicate with a portion of the budget, and then allocates the budget accordingly to improve the accuracy of the model.
    \item We devise a sequential predicate selection solution, SPS, which adaptively conducts exploration and exploitation, by iteratively requesting a small number of records in each query, aiming to improve the predicate utility estimates and utilize such estimates at the same time. 
    \item We design methods to estimate the expected utility of a given predicate, allowing EA and SPS to proceed without necessitating the re-training of the underlying model. 
    \item We experimentally study the proposed solutions across a variety of settings, including but not limited to, different tasks, ML models, datasets, distributions, budget limitations, utility measures. We showcase that each solution has certain benefits and can be suitably adopted when applied to real-world settings.
\end{itemize}
The rest of this paper is organized as follows.  In Section 2, we review related work. Section 3 introduces terminology and background. In Section 4, we propose the Estimation-and-Allocation method, followed by Section 5 that presents the Sequential Predicate Selection approach. The experimental evaluation is presented in Section 6. Section 7 concludes this paper.
\section{Related Work}
\label{sec:related}

\textbf{Data Markets.} Recently Fernandez et al. \cite{fernandez2020data} presented their vision for the design of platforms to support data markets. Our work is aligned with such a vision addressing a specific problem in this setting. Data markets have been an active research topic in various communities. For example, several works adopt a game-theoretic approach of market design for such markets \cite{DBLP:conf/ec/AgarwalDS19,DBLP:conf/ec/MehtaDJM19,liu2020dealer}, exploring issues such as fairness and pricing. 

\textbf{Data Pricing.} A lot of research has been devoted to the design of pricing mechanisms for data. A recent survey \cite{DBLP:conf/kdd/Pei20} presents an overview of works related to data pricing from a data science perspective. Other works \cite{DBLP:journals/pvldb/ChawlaDKT19,lin2014arbitrage,balazinska2011data,DBLP:conf/sigmod/ChenK019}, present pricing-specific problems for queries and models. Pricing mechanisms for data are an orthogonal problem to our work. Any applicable pricing mechanism can be adopted to determine the number of records one can obtain based on applicable monetary budgets.

\textbf{Data Acquisition.} Another area related to data markets is the acquisition of data. Chen et al. \cite{DBLP:conf/sigecom/ChenILSZ18} consider the problem of obtaining data from multiple data providers in order to improve the accuracy of linear statistical estimators. The focus is strictly on linear statistical techniques and they provide certain types of guarantees on the best strategies to adopt when providers decide to abstain from making their data available and thus data costs vary dynamically. In a related thread Kong et al. \cite{DBLP:conf/aaai/KongSTY20} study the problem of estimating Gaussian parameters and other estimators with Gaussian noise. Zhang et al. \cite{zhang2020optimal} study incentive mechanisms for participants in data markets to refresh their data.

\textbf{Data Augmentation.} Data augmentation and feature selection for machine learning have also attracted research interests recently \cite{DBLP:journals/pvldb/ShahKZ17,kumar2016join,DBLP:journals/pvldb/ChepurkoMZFKK20}. Kumar et al. \cite{kumar2016join} focus on ML model training on relational data and propose methods to decide whether joining certain attributes would improve the model's accuracy. Chepurko et al. \cite{DBLP:journals/pvldb/ChepurkoMZFKK20} design a system that automatically performs feature selections and joins so as to improve the performance of a given model. Their work, however, focuses on the selection and acquisition of (new) features rather than records, and assumes the accessibility of all features, and thus is orthogonal to the problem studied herein.

\textbf{Active Learning.} Active learning is concerned with interactively acquiring labels for new data points to improve the performance of machine learning models \cite{pmlr-v108-shui20a}. It has been applied to solve various problems in data management \cite{lin2020activelearning_sigmod}. In a typical setting for active learning, we have access to the features of new data and have to decide the set of records for which we would like to acquire the label for a cost. In contrast, we deal with the scenario where the consumer does not have any prior knowledge of the provider's data other than the meta-data. Perhaps what is most related to ours in this area is the work by Lomasky et al. \cite{Lomasky2007ECML}, which also aims to control the class of data to acquire for improving the ML models. However, they focus on the task of selecting class proportions for generating new training data. They do not explicitly optimize the use of a fixed budget, do not tackle regression problems, and require frequent re-training of the model.


\textbf{Exploration vs. Exploitation.} The trade-off between exploration and exploitation exists in many scenarios where a decision has to be made with incomplete knowledge. Although methods balancing this trade-off have been proposed in various contexts, such as reinforcement learning \cite{sutton2018reinforcement} and online decision making \cite{ikonomovska2015real}, they are not directly applicable to our problem setting. Nonetheless, we share a similar methodology and adopt one of the popular approaches for performing this trade-off, Thompson Sampling \cite{russo2017tutorial}, as the framework to build the proposed SPS solution. 


 

\section{Preliminaries}
\label{sec:preliminaries}
In this section, we formally define the terminology utilized and the problems we focus on in this paper.



{\bf Data Domain and Learning Task.} We consider a supervised learning task defined on a data domain $\Gamma=\{\mathcal{X},\mathcal{Y}\}$, where $\mathcal{X}$ denotes the {\em feature} space and $\mathcal{Y}$ denotes the {\em label} space (e.g., all possible class labels for classification tasks, all possible values of the dependent variable for regression tasks). Suppose there is a conditional distribution  $p(y|x)$ defined over $\Gamma$, where $x \in \mathcal{X}$, $y \in \mathcal{Y}$. 
The learning task is to train a model $\mathcal{M}$ on a dataset that represents a distribution $g$ that is as close as possible to the target distribution $p$. The accuracy of $\mathcal{M}$ is evaluated on a testing dataset $\mathcal{D}_{test}\subset\Gamma$.  In accordance to any  well-formed learning task, we assume that both the training data and testing data come from the same distribution $p$. The model $\mathcal{M}$ is evaluated using a function $F$ based on $\mathcal{D}_{test}$, denoted as $F(\mathcal{M};\mathcal{D}_{test})$. 

{\bf Provider, Consumer, and Budget.} The provider maintains a collection of data records $\mathcal{D}_{pool}\subset\Gamma$ drawn from the target distribution $p$, which are provided in the data market and are initially entirely invisible to the consumer. 
The consumer has an ML model $\mathcal{M}$ trained on data $\mathcal{D}_{init} \subset \Gamma$ drawn from $p$. Note that although both $\mathcal{D}_{pool}$ and $\mathcal{D}_{init}$ follow the same distribution $p$, they are not necessarily representative samples of $p$. For example, $\mathcal{D}_{pool}$ may contain a high percentage of records from one part of the data domain $\Gamma$, while $\mathcal{D}_{init}$ from another. 
In the degenerate case, $\mathcal{M}$ is simply a model randomly initialized without using any training data, i.e., $\mathcal{D}_{init}=\emptyset$. The consumer has a budget $B$, which is the maximum number of records that the consumer can obtain from the provider.

{\bf Query and Predicate.} A query $Q = (P, I)$ consists of a predicate $P$ that specifies the properties of the records the consumer would like to acquire, and an integer $I$ that specifies the number of records requested from the provider. Let $\mathcal{R}_P\subset\mathcal{D}_{pool}$ denote the set of records that satisfy  $P$ and $\mathcal{R}_Q$ denote the set of records returned by the provider. All possible predicates admissible to the provider constitute set $\mathcal{P}$, which can be formed in various ways depending on the task at hand. 
{In the paper we adopt a simple yet intuitive predicate construction strategy: for a classification problem, $\mathcal{P}$ contains all predicates with a selection on the class label (e.g., $\texttt{label = `dog'}$); for a regression problem, we discretize each attribute into equal-width sub-ranges, and all combinations of sub-ranges, denoted by ``cells'', constitute $\mathcal{P}$ (e.g., \texttt{1000$\le$ salary $\le$ 2000 $\wedge$ 20$\le$ age$\le$ 30}). There are many other strategies to construct and refine $\mathcal{P}$. For example, the consumer may perform cross validation on $\mathcal{D}_{init}$ and use the $m$ labels in which $\mathcal{M}$ has the lowest accuracy to construct $\mathcal{P}$ for more targeted acquisition. One may also inject domain knowledge to the predicate construction process and only use these classes or cells related to the task as predicates. For example, a consumer training a cat/dog classifier may not consider predicate $\texttt{label = `horse'}$.  Refer to Section \ref{exp:settings} and Section \ref{exp:partition} for more details on the methodology to construct $\mathcal{P}$ adopted in this work and the associated evaluation.
We note that the investigation of predicate construction strategies and their properties is an interesting direction for future work. Our emphasis is to develop methods for data acquisition that are independent of the predicate construction process.}

{\bf Interaction.} Each round of interaction between the consumer and the provider consists of two steps: (1) the consumer issues a query $Q=(P,I)$ to the provider, and (2) the provider returns a set of records $\mathcal{R}_Q$, where $\vert\mathcal{R}_Q\vert=I$ and each $r\in\mathcal{R}_Q$ is randomly sampled from $\mathcal{R}_P$ {\footnote{Note that if $I$ is larger than the number of the provider's remaining records (say $I_P$), all of the provider's records will be returned and only $I_P$ will be deducted from the consumer's budget.}}. Without loss of generality, we assume that all records provided to the consumer are unique within the same and across different rounds of interactions. {A predicate $P$ can be reused across different rounds of interactions as long as $\mathcal{R}_P$ has not been exhausted, i.e., there are records in $\mathcal{R}_P$ that have not been acquired by the consumer yet.}  

{{\bf Predicate utility.} 
After each round of interaction, the consumer estimates the utility of the predicate used in the query, which is useful for planning the next round of interaction. The utility of a predicate $P$ expresses the anticipated accuracy improvement that $\mathcal{R}_P$ brings to $\mathcal{M}$. We define a measure that we call {\em novelty} to quantify predicate utility. The basic idea of this measure is to quantify the difference between the data acquired in the interaction to those that the consumer currently possesses. The higher the difference, the more information this interaction brings to the consumer. {Let $\mathcal{R}_{\mathcal{M}:P}$ be the records the consumer currently possesses satisfying $P$. The novelty of predicate $P$, denoted as $U_P$, is defined on $\mathcal{R}_{\mathcal{M}:P}$ and  $\mathcal{R}_P$. More specifically, we consider a binary classification problem that treats {$\mathcal{R}_{\mathcal{M}:P}$ and  $\mathcal{R}_P$ as samples from class 0 and class 1 respectively}, and train a classifier CLF to distinguish between the two sets of records. The utility of $P$ is computed as follows:

\begin{equation}
    \label{eq:clf}
    U_P=\frac{\sum_{(x,y)\in\mathcal{R}_P}\mathbb{I}[\mbox{CLF}((x,y))=1]}{|\mathcal{R}_P|}
\end{equation}
where $\mathbb{I}[*]$ is the indicator function that takes value 1 if statement $*$ is true and 0 otherwise, and $\mbox{CLF}((x,y))$ denotes the prediction for record $(x,y)$ made by CLF. In principle, any classifier may be used as the CLF. However, in practice it is preferred to use light-weight models for faster training and inference as the computation of novelty is carried out frequently. We study the influence of different classifiers in Section \ref{sec:exp-clf}.}

The intuition for the design of novelty is that, if $\mathcal{R}_P$ is drawn from a distribution that is very different than the one $\mathcal{R}_{\mathcal{M}:P}$ is drawn from, then the two sets of records can be easily differentiated and $U_P$ is high, and vice versa. Note that we only evaluate the accuracy of the classifier on $\mathcal{R}_P$, because novelty measures how different $\mathcal{R}_P$ is, given $\mathcal{R}_{\mathcal{M}:P}$, rather than the other way around. Such methods for quantifying the difference between two distributions are well adopted in the ML literature \cite{DBLP:conf/iclr/Lopez-PazO17}. In practice, it may not be feasible for the consumer to obtain all the records in $\mathcal{R}_P$ due to budget limitations. As such, in the proposed solutions, we utilize queries based on the same predicate $P$ returning $|\mathcal{R}_Q| \ll|\mathcal{R}_P|$ records, to estimate $U_P$. 

}

{\bf Acquisition plan.} The acquisition plan of the consumer consists of a sequence of interactions, $\langle(P_1,I_1),(P_2,I_2),\cdots,(P_z,I_z)\rangle$, where $\forall i\in[1,z]$, $P_i\in\mathcal{P}$ and $\sum_{i=1}^z I_i=B$. 
The consumer receives $B$ records in total after executing the acquisition plan, denoted as $\mathcal{D}_{otd}$.

The problem of data acquisition for model improvement is defined as follows.
\begin{definition}\textbf{Data Acquisition for Model Improvement}.\\ Given (1) a set of records,  $\mathcal{D}_{pool}$ from a provider, (2) an initial set of records, $\mathcal{D}_{init}$, possessed by a consumer, (3) the set of possible predicates, $\mathcal{P}$, (4) the initial model, $\mathcal{M}$, of the consumer, (5) a measure to evaluate model accuracy, $F$, and (6) the budget, $B$, the objective of data acquisition for model improvement is to construct an acquisition plan to maximize  $F(M';\mathcal{D}_{test})$, where $M'$ denotes the consumer model $\mathcal{M}$ after being re-trained on $\mathcal{D}_{init}\cup\mathcal{D}_{otd}$.
\end{definition}

We summarize the notations and abbreviations used in the paper in Appendix \ref{sec:notation}.
\section{An Estimation-and-Allocation Solution}
\label{sec:ea}
In this section, we introduce {\em estimation-and-allocation} (EA), a two-stage solution, to generate the acquisition plan. Essentially, stage one of EA is designed to explore, i.e., to gather more information on how useful each predicate is; while stage two is to exploit, i.e., to utilize the knowledge gained in stage one to optimize subsequent actions. Specifically, the first stage, called the {\em Estimation Stage}, aims to obtain accurate estimates on the utilities of predicates in $\mathcal{P}$; this is achieved via querying the provider requesting a number of records that constitute a small portion of the budget. Then in the second stage, called the {\em Allocation Stage}, the consumer allocates the remaining budget and issues queries to the provider based on the estimated predicate utilities.  We first discuss in Section~\ref{sec:bounding_error} how to ensure the quality of the utility estimates in the Estimation Stage, and present a method that could balance between quality and budget consumption (the amount of record budget spent) in Section~\ref{sec:adaptive_estimation}. We then elaborate on the Allocation Stage in Section~\ref{sec:allocation}. 

\subsection{Estimating Predicate Utility}
\label{sec:bounding_error}
We now discuss how to estimate the predicate utilities. Recall that the utility $U_P$ of a predicate $P$ is defined as the accuracy of the classifier (denoted by CLF) in  differentiating $\mathcal{R}_{P}$ from $\mathcal{R}_{\mathcal{M}:P}$. Since it is not possible to obtain the entire $\mathcal{R}_P$, we rely on queries using predicate $P$ to effectively sample from it. 
We can estimate $U_P$ based on the records already acquired with $P$, say $\hat{\mathcal{R}}_P$, and we use $\hat{U}_P$ to denote the estimated value of $U_P$. 

The effectiveness of EA depends largely on the accuracy of the predicate utility estimates; thus we investigate how to statistically bound the estimation error, i.e., the difference between $U_P$ and $\hat{U}_P$. To this end, we aim to find the $\epsilon$-$\delta$ approximations of all predicate utilities, defined as follows.

\begin{definition}\textbf{$\epsilon$-$\delta$ Approximation}. $\hat{U}_P$ is said to be an $\epsilon$-$\delta$ approximation of $U_P$ if $\mbox{Pr}(\vert U_P-\hat{U}_P\vert\ge \epsilon)\le\delta$, where $\epsilon$ denotes the error bound and $\delta$ denotes the significance level.
\end{definition}

In order to determine whether $\hat{U}_P$ is an $\epsilon$-$\delta$ approximation of $U_P$, we conduct a statistical test with the following null hypothesis:

\begin{equation}
    \label{eq:null}
    H_0^{(P)}: \vert U_P-\hat{U}_P\vert\ge \epsilon
\end{equation}

where $P\in\mathcal{P}$ is an arbitrary predicate. We reject $H_0^{(P)}$ at significance level $\delta$ when the following condition is met:

\begin{equation}
    \label{eq:reject}
    \mbox{reject}\ H_0^{(P)}\ \mbox{if}\ \mbox{Pr}(\vert U_P-\hat{U}_P\vert\ge \epsilon)\le\delta
\end{equation}

The rejection condition bounds the probability of type I error (false rejection) of the statistical test by $\delta$, and if $H_0^{(P)}$ can be rejected, clearly $\hat{U}_P$ is an $\epsilon$-$\delta$ approximation of $U_P$. 

We next discuss how to compute $\mbox{Pr}(\vert U_P-\hat{U}_P\vert\ge \epsilon)$. {Note that there are two sources of error in estimating $U_P$: (1) approximating $U_P$ with a subset $\hat{\mathcal{R}}_P\subset{\mathcal{R}}_P$, and (2) the error incurred by the CLF. In our work, we focus on the error caused by insufficient records (i.e., source (1)), and we bound the model error (i.e., source (2)) following \cite{dutt2020efficiently}. Nonetheless,  both types of error can be reduced by increasing the size of $\hat{\mathcal{R}}_P$ \cite{dutt2020efficiently}. To bound the error attributed to insufficient records, we present the following result on the distribution of $U_P-\hat{U}_P$.}
\begin{theorem}
\label{theorem:normal}
$U_P-\hat{U}_P\sim \mathcal{N}(0,\frac{U_P(1-U_P)}{\vert\hat{\mathcal{R}}_P\vert})$, where $\mathcal{N}(0,\frac{U_P(1-U_P)}{\vert\hat{\mathcal{R}_P}\vert})$ denotes the normal distribution with mean $0$ and variance $\frac{U_P(1-U_P)}{\vert\hat{\mathcal{R}_P}\vert}$.
\end{theorem}

\begin{proof}
Since $U_P$ is the accuracy of the binary classifier CLF in discriminating $\mathcal{R}_P$ from $\mathcal{R}_{\mathcal{M}:P}$, for each record $(x_i,y_i)\in\mathcal{R}_P$, either CLF($(x_i,y_i)$)$=0$ or CLF($(x_i,y_i)$)$=1$, corresponding to the case when $(x_i,y_i)$ is regarded by CLF to be from $\mathcal{R}_P$ or $\mathcal{R}_{\mathcal{M}:P}$, respectively. Also, if we let  $r_i=\mathbb{I}[\mbox{CLF}((x_i,y_i))=0]$, $r_i$ can be viewed as an independent Bernoulli($p$) variable, with $p$ being the probability of $r_i=1$. Evidently, in this case, $p=U_P$.

As $\hat{U}_P=\frac{1}{\vert\hat{\mathcal{R}}_P\vert}\sum_{(x_i,y_i)\in\hat{\mathcal{R}}_P}r_i$ and $r_i\sim\mbox{Bernoulli}(U_P)$, we know $\hat{U}_P*\vert\hat{\mathcal{R}}_P\vert\sim\mbox{Binomial}(\vert\hat{\mathcal{R}}_P\vert, U_P)$. According to the Central Limit Theorem, we have $U_P-\hat{U}_P\sim \mathcal{N}(0,\frac{U_P(1-U_P)}{\vert\hat{\mathcal{R}}_P\vert})$.
\end{proof}

Since $U_P(1-U_P)$ is not known, we cannot directly estimate the difference between $U_P$ and $\hat{U}_P$ based on Theorem \ref{theorem:normal}. Following the standard practice \cite{ash2000probability} in estimating population mean ($U_P$ in our case), we introduce statistic $t_P$ as follows:

\begin{equation}
    \label{eq:tp}
    t_P=\frac{U_P-\hat{U}_P}{S_P/\sqrt{\vert\hat{\mathcal{R}}_P\vert}}
\end{equation}

where $S_P=\sqrt{\hat{U}_P(1-\hat{U}_P)}$ denotes the sample standard deviation. Clearly $t_P$ follows $t$-distribution with degree of freedom ($\vert\hat{\mathcal{R}}_P\vert-1$). 

Now we can re-write $U_P-\hat{U}_P$ as follows:

\begin{equation}
    U_P-\hat{U}_P={t_P*S_P}/{\sqrt{\vert\hat{\mathcal{R}}_P\vert}}
\end{equation}

The probability of $\vert U_P-\hat{U}_P\vert\ge \epsilon$ can now be computed as follows:

\begin{equation}
\begin{split}
    \mbox{Pr}(\vert U_P-\hat{U}_P\vert\ge \epsilon)&=\mbox{Pr}\left(\vert{t_P*S_P}/{\sqrt{\vert\hat{\mathcal{R}}_P\vert}}\vert\ge\epsilon\right)\\
    &=\mbox{Pr}\left(\vert t_P\vert\ge{\epsilon\sqrt{\vert\hat{\mathcal{R}}_P\vert}}/{S_P}\right)\\
    &\le\int_{-\infty}^{-\frac{\epsilon\sqrt{\vert\hat{\mathcal{R}}_P\vert}}{S_P}}f_{n_P}(t_P)dt_P + \int_{\frac{\epsilon\sqrt{\vert\hat{\mathcal{R}}_P\vert}}{S_P}}^{\infty}f_{n_P}(t_P)dt_P\\
    &=Z_P
\end{split}
\label{eq:zr}
\end{equation}

where $n_P=\vert\hat{\mathcal{R}}_P\vert-1$, and $f_{n_P}$ is the probability density function of the $t$-distribution with degree of freedom $n_P$.

Therefore, $\mbox{Pr}(\vert U_P-\hat{U}_P\vert\ge \epsilon)=Z_P$, and $H_0^{(P)}$ can be rejected if $Z_P\le\delta$. Evidently $Z_P$ is negatively correlated to $\vert\hat{\mathcal{R}}_P\vert$, and thus if $H_0^{(P)}$ cannot be rejected, one can obtain more records utilizing predicate $P$ and issuing additional queries to reduce the value of $Z_P$ until $Z_P\le\delta$. The intuition behind this process is that, the more records the consumer obtains utilizing $P$, the more accurate the estimate $\hat{U}_P$ is, and the smaller $\vert U_P-\hat{U}_P\vert$ would be. When the null hypotheses for all predicates in $\mathcal{P}$ can be rejected, the current predicate utility estimates are $\epsilon$-$\delta$ approximations.

\subsection{Budget-Aware Utility Estimation}
\label{sec:adaptive_estimation}
The Estimation Stage has to consider two conflicting goals: providing more accurate estimation of predicate utilities and controlling the budget consumption so that there is more budget left to spend in the Allocation Stage. We thus introduce a budget-aware estimation method, which first acquires a small number of records using each predicate in $\mathcal{P}$ and computes utility estimates accordingly, and then iteratively determines for each predicate whether obtaining more records to improve the estimation accuracy is worthwhile based on a measure called {\em heuristic reward}. Such a measure is designed to strike a balance between estimation accuracy and budget consumption.

Since for a given significance level $\delta$, the estimation accuracy is contingent on $\epsilon$, we adaptively choose and adjust its value in order to yield estimates with different levels of accuracy. Intuitively, to achieve higher estimation accuracy, i.e., smaller $\epsilon$, the consumer needs to acquire more records for estimation. Assume that the consumer has acquired records $\hat{\mathcal{R}}^0_P$ for each predicate $P\in\mathcal{P}$. We use $B'=B-\sum_{P\in\mathcal{P}}\vert\hat{\mathcal{R}}^0_P\vert$ to denote the remaining budget of the consumer. Let $\epsilon_0$ be the minimal $\epsilon$ that causes the null hypotheses for all predicates in $\mathcal{P}$ to be rejected (note that such $\epsilon_0$ always exists, with the extreme case being $\epsilon_0=1$). The heuristic reward is defined as $B'\cdot(1-\epsilon_0)$, which is larger if (1) $B'$ is large, meaning the consumer has more available budget, and (2) $\epsilon_0$ is small, meaning that the estimations are accurate.

\begin{example}
\label{example:ea-init}
Consider a consumer with budget=500 and there are 5 predicates to choose from. Assume that the consumer has acquired 5 records for each predicate, and the resulting sample standard deviations are $[0.1,0.11,0.12,0.13,0.14]$ respectively. Let the significance level $\delta$ be 0.01. Using Equation (\ref{eq:zr}), we know the minimal values of $\epsilon$ causing all $H_0^{(P)}$ ($P\in\mathcal{P}$) to be rejected are $[0.21,0.23,0.25,0.27,0.29]$, and thus the $\epsilon$ that causes all null hypotheses to be rejected, or $\epsilon_0$, is 0.29. Since the remaining budget is 475, the heuristic reward is thus $475\cdot(1-0.29)=337.25$.
\end{example}

Now assume the consumer aims to determine whether reducing $\epsilon_0$ to $\epsilon_b$, by acquiring more records, would improve the heuristic reward. In order to compute the heuristic reward corresponding to $\epsilon_b$, we need to estimate how many additional records need to be acquired, $\Delta B_b$, to reach $\epsilon_b$. We next show how to estimate the value of $\Delta B_b$.

Recall from Equation (\ref{eq:zr}) that the reject condition of $H_0^{(P)}$ is $Z_P\le\delta$. Since $\vert\hat{\mathcal{R}}_P\vert$ is negatively correlated to $Z_P$, in order to get the minimal number of records to acquire to reach a given $\epsilon$, we use the maximal $Z_P$, i.e., $Z_P=\delta$, and rewrite Equation (\ref{eq:zr}) as follows:

\begin{equation*}
\begin{split}
    Z_P=\delta\Leftrightarrow&\int_{-\infty}^{-\frac{\epsilon\sqrt{\vert\hat{\mathcal{R}}_P\vert}}{S_P}}f_{n_P}(t_P)dt_P + \int_{\frac{\epsilon\sqrt{\vert\hat{\mathcal{R}}_P\vert}}{S_P}}^{\infty}f_{n_P}(t_P)dt_P=\delta\\
    \Leftrightarrow&\int_{-\infty}^{-\frac{\epsilon\sqrt{\vert\hat{\mathcal{R}}_P\vert}}{S_P}}f_{n_P}(t_P)dt_P=\frac{\delta}{2}
\end{split}
\label{eq:zp-rewrite}
\end{equation*}

Let $A_P=\frac{\epsilon\sqrt{\vert\hat{\mathcal{R}}_P\vert}}{S_P}$, we can further rewrite the equation above as follows: 

\begin{equation}
    \label{eq:ap}
    \int_{-\infty}^{-A_P}f_{n_P}(t_P)dt_P=\frac{\delta}{2}\Leftrightarrow A_P=-\mbox{PCT}_{n_P}(\frac{\delta}{2})
\end{equation}

where PCT$_{n_P}$ denotes the percentile function of $t$-distribution with degree of freedom $n_P$.

Now having a way to determine the value of $A_P$ using Equation (\ref{eq:ap}), {we rewrite $A_P=\frac{\epsilon\sqrt{\vert\hat{\mathcal{R}}_P\vert}}{S_P}$ as follows:} 

\begin{equation}
\label{eq:relation}
    \vert\hat{\mathcal{R}}_P\vert = \left(\frac{A_P\cdot S_P}{\epsilon}\right)^2
\end{equation}

Equation (\ref{eq:relation}) establishes the relation between the number of records currently obtained and the error bound $\epsilon$ that can be obtained using those records.

Let $\hat{\mathcal{R}}_P^b$ be the records with which we can reach error bound $\epsilon_b$, we have:

\begin{equation}
\label{eq:eps_b}
    \vert\hat{\mathcal{R}}_P^b\vert = \left(\frac{A_P^b\cdot S_P^b}{\epsilon_b}\right)^2
\end{equation}

where $S_P^b$ denotes the sample standard deviation of $\hat{\mathcal{R}}_P^b$, and $A_P^b=-\mbox{PCT}_{n_P^b}(\frac{\delta}{2})$, $n_P^b=\vert\hat{\mathcal{R}}_P^b\vert-1$. 
Notice that $S_P$ asymptotically converges to the population standard deviation, and $A_P^b$, which is determined by PCT$_{n_P}$, asymptotically converges to the opposite value of the $\frac{\delta}{2}$-percentile of standard normal distribution \cite{fisher1992statistical}. Thus, although $S_P^b$ and $A_P^b$ are unknown, we choose to approximate their values by $S_P^0$ and $A_P^0$, i.e., the values computed based on $\hat{\mathcal{R}}_P^0$, as long as $|\hat{\mathcal{R}}_P^0|$ is reasonably large. We demonstrate the validity of this approximation empirically as well in Section \ref{exp:l}.
Therefore, we use  $\left(\frac{A_P^0\cdot S_P^0}{\epsilon_b}\right)^2$ as the least number of records required to achieve estimation error $\epsilon_b$ for predicate $P$.
\begin{example}
\label{example:delta-b}
Following Example \ref{example:ea-init}, we assume that the consumer plans to reduce the error bound to 0.15 from 0.29. The values of $A_P^0$ for each $P$ computed with Equation \ref{eq:ap} are: $[4.68,4.66,4.64,4.63,4.62]$. Thus the number of records required for each predicate to reach error bound 0.15 are: $[10,12,14,17,19]$.
\end{example}


 Since we already posses $\vert\hat{\mathcal{R}}_P^0\vert$ records satisfying $P$, we need to acquire $\left(\frac{A_P^0\cdot S_P^0}{\epsilon^b}\right)^2-\vert\hat{\mathcal{R}}_P^0\vert$ more records utilizing predicate $P$. The total additional records from our budget needed to reach $\epsilon_b$ from $\epsilon_0$ is thus
 
\begin{equation}
\label{eq:mb}
    \Delta B_b=\sum_{P\in\mathcal{P}}\left[\left(\frac{A_P^0\cdot S_P^0}{\epsilon_b}\right)^2-\vert\hat{\mathcal{R}}_P^0\vert\right].
\end{equation}

The new heuristic reward after obtaining these $\Delta B_b$ records can thus be estimated as $(B'-\Delta B_b)\cdot(1-\epsilon_b)$. 
Let $\epsilon_*=\arg\max_{\epsilon_b}(B'-\Delta B_b)\cdot(1-\epsilon_b)$, be the value $\epsilon_b$ yielding the maximal heuristic reward. The consumer then compares $B'\cdot(1-{\epsilon_0})$ with $(B'-\Delta B_{\epsilon_{*}})\cdot(1-\epsilon_*)$. If the latter value is higher, meaning that the new combination of the estimation error $\epsilon_*$ and remaining budget ($B'-\Delta B_*$) is better, we initiate a new interaction to obtain $\Delta B_*$ more records according to Equation (\ref{eq:mb}). This process continues until the above calculation indicates no more improvement in heuristic reward is possible via further interaction; we then terminate the Estimation Stage and enter the Allocation Stage.
\subsection{Budget Allocation}
\label{sec:allocation}
The Allocation Stage of EA is concerned with distributing the remaining budget across all predicates in $\mathcal{P}$ utilizing their estimated utilities. We consider two budget allocation strategies, where $M_P$ denotes the budget (number of records) allocated to a predicate $P$:
\begin{itemize}
    \item {\bf Linear Allocation}: $M_P=\frac{B*\hat{U}_P}{\sum_{P'\in\mathcal{P}}\hat{U}_P'}-B_P$, where $B_P$ denotes the number of records obtained with $P$ during the Estimation Stage.
    \item {\bf Square-root Allocation}:  $M_P=\frac{B*\sqrt{\hat{U}_P}}{\sum_{P'\in\mathcal{P}}\sqrt{\hat{U}_P'}}-B_P$.
\end{itemize}


We sort all predicates in descending order of their utilities and sequentially obtain records based on $M_P$ starting from the predicate with the highest estimated utility. This stage continues iteratively until the budget is exhausted. We demonstrate in Section \ref{exp:allocation} that each allocation strategy has its own winning cases and therefore the allocation strategy can be selected based on the specific task. 
\subsection{The EA Algorithm}
\label{sec:ea_alg}
The  EA solution is summarized in Algorithm \ref{alg:ea}. We first use $l\%$ ($l$ is configurable) of the budget to acquire records using each predicate to start the estimation (Line 1). We calculate the current estimation quality $\epsilon_0$ (Line 4) and the estimation quality that is expected to yield the highest heuristic reward $\epsilon_*$ (Line 5). The current heuristic reward is compared with the expected highest heuristic reward (Line 7), and if the former is higher, we terminate the Estimation Stage (Lines 7-8) and enter the Allocation Stage (Lines 12-13); otherwise we obtain more records based on $\epsilon_*$ (Lines 10-11) and repeat the process. Note that the Estimation Stage also terminates when the budget is exhausted. However although Line 3 provides an exit to the estimation stage, corresponding to the case when all budget is consumed during estimation, such case will never happen as exhausting all budget provides a heuristic reward of 0 and thus is prohibited by Line 7.

Since during the estimation stage more records can be acquired, it is suggested to initialize Algorithm \ref{alg:ea} with a small value of $l$, as initialization with a large value of $l$ may consume too much budget. However, if $l$ is too small (say only 1 record for each predicate), the sample standard deviation $S_P$ computed may accidentally be zero and as a result, $Z_P$ is zero too (Equation (\ref{eq:zr})). Consequently, $H_o^{(P)}$ can be rejected with any $\epsilon$ (even zero) because $Z_P\le\delta$ is always true, and the estimation stage terminates immediately and abnormally as $\epsilon$ cannot be further reduced. With these trade-offs in mind, we experimentally study the influence of the choice of $l$ in Section \ref{exp:l}.
\begin{algorithm}[htb]
\caption{Estimation-and-Allocation}
\label{alg:ea}
\begin{algorithmic}[1]
\Require
budget $B$, all predicates $\mathcal{P}$
\State \textbf{Initialization:} for each $P\in\mathcal{P}$, acquire $l\%$ random records in $\mathcal{R}_P$; assume remaining budget is $B'$.
\State //Estimation Stage
\While {$B'>0$}
\State $\epsilon_0\leftarrow$ minimal $\epsilon$ to cause the hypotheses to be rejected;
\State $\epsilon_*\leftarrow\arg\max_{\epsilon_b,\epsilon_b<\epsilon_0}(B'-\Delta B_{\epsilon_b})\cdot(1-\epsilon_b)$;
\State //$\Delta B_{\epsilon_b}$ is computed with Equation (\ref{eq:mb})
\If {$B'\cdot(1-\epsilon_0)\ge(B'-\Delta B_{\epsilon_*})\cdot(1-\epsilon_*)$}
\State break;
\Else
\State acquire $\Delta B_{\epsilon_b}'$ more records according to Equation (\ref{eq:mb});
\State $B'=B'-\Delta B_{\epsilon_b}$;
\EndIf
\EndWhile
\State //Allocation Stage
\State Allocate the remaining budget using strategies in Section \ref{sec:allocation};
\end{algorithmic}
\end{algorithm}
\newcommand{\pr}{\textrm{Pr}}
\section{A Sequential Predicate Selection Solution}
\label{sec:sps}
In this section, we adopt a Bayesian probabilistic approach and introduce an alternative solution called {\em Sequential Predicate Selection} (SPS). While EA employs two separate stages for exploration and exploitation, SPS utilizes many rounds of interactions, acquiring a small number of records in each interaction with varying predicates, achieving both exploration and exploitation in the same round. In particular, in each round, SPS balances between two objectives: (1) acquiring more records using predicates that are expected to provide higher accuracy improvement to the model, based on previous observations; and (2) exploring to identify other predicates that may bring even higher accuracy improvement to model accuracy. We implement the design utilizing Thompson Sampling (TS), an action selection method with proven performance \cite{russo2017tutorial,ikonomovska2015real}.


\subsection{Framework of Thompson Sampling}

{Thompson Sampling \cite{russo2017tutorial} proceeds as follows. Given an action space $\mathcal{A}$, an agent conducts actions $a_1,a_2,\cdots$, each selected from $\mathcal{A}$, in rounds. After applying $a_t$ {in round $t$}, the agent observes a reward $r_t$, which is randomly generated according to a conditional probability measure $q_\theta(\cdot|a_t)$. The agent is initially uncertain about the value of {parameters} $\theta$ and thus uses a prior distribution $p$ to describe $\theta$, which is iteratively updated based on ($a_t,r_t$) pairs. The target is to maximize the cumulative rewards {over a given number of rounds}. Given $\mathcal{A}$, $p$, and $q$, TS repeats the following steps for action selection:
\begin{enumerate}
    \item Sample $\hat{\theta}\sim p$, i.e., {sample the parameters} controlling the reward according to $p$ ($p$ is called the prior distribution of $\theta$ in this round).
    \item Let $a_t=\arg\max_{a\in\mathcal{A}}\mathbb{E}_{q_{\hat{\theta}}}[r_t|a_t=a]$, i.e., $a_t$ is the action {with the maximum  expected reward under parameter values $\hat{\theta}$}. Perform action $a_t$ and observe $r_t$.
    \item Update $p=\mathbb{P}_{p,q}(\theta\in\cdot|a_t,r_t)$, i.e., $p$ becomes the posterior distribution of $\theta$ given ($a_t,r_t$).
\end{enumerate}

In the following, we apply TS to the problem of predicate selection, discuss the choices of $p$ and $q$ for the problem, and develop the update method of $p$ from prior distribution to posterior distribution. In addition, as will be shown in Section \ref{sec:nonstationary}, the rewards of predicates evolve over time in our problem, and thus we also design methods to handle changes in $\theta$.
}
\subsection{Sequential Predicate Selection Using Thompson Sampling}
\label{sec:spsts}
In SPS, we repeatedly issue queries to the provider in multiple rounds of interactions, until the budget $B$ is exhausted. In each round, a query $Q=(P,I)$ is issued and we calculate a value called {\em query reward} for the records received, deciding on the next query to generate based on the knowledge acquired so far. The objective of the consumer is to maximize the cumulative reward for the queries issued in all rounds of interactions. In what follows, we detail the set of queries that the consumer can issue, how query reward is evaluated, and how knowledge regarding the problem space (queries and the resulting rewards) is updated. 

The number of possible queries the consumer may ask is $\vert\mathcal{P}\vert\cdot B$, as each pair of $P\in\mathcal{P}$ and $I\in[1, B]$ can form a query. We assume that a fixed $I$, denoted by $I_\Delta$, is used for each query. Thus, choosing a query boils down to selecting a predicate in $\mathcal{P}$, and in the sequel we use the term {\em predicate reward} of $P$ and the query reward of $Q=(P,I_\Delta)$ interchangeably. We empirically study the influence of $I_\Delta$ in Section \ref{exp:batch-size}.

{Let $\mathcal{R}_P^{I_\Delta}$ be the records returned by query $Q=(P,I_\Delta)$. The computation of query reward follows the spirit of predicate utility, i.e., novelty introduced in Section \ref{sec:preliminaries}, except for that instead of using the CLF to differentiate $\mathcal{R}_P$ from $\mathcal{R}_{\mathcal{M}:P}$, the query reward differentiates $\mathcal{R}_P^{I_\Delta}$ from $\mathcal{R}_{\mathcal{M}:P}$, directly measuring the anticipated accuracy improvement $\mathcal{R}_P^{I_\Delta}$ brings to model $\mathcal{M}$.} More specifically, the reward of $Q$ is the number of records in $\mathcal{R}_P^{I_\Delta}$ that can be correctly labelled by CLF, following the developments in Section~\ref{sec:preliminaries}. 
Since records in $\mathcal{R}_P^{I_\Delta}$ are randomly selected from $\mathcal{R}_P$, the reward $r$ of $P$ can be viewed as a random variable, which is assumed to follow a distribution
$\pr(r|\theta_P,P)$, where $\theta_P$ refers to the set of parameters of this distribution. In each interaction, the consumer randomly draws a value $r_P$  from $\pr(r|\theta_P,P)$ for each $P\in\mathcal{P}$, selects the predicate $P^*=\arg\max_{P\in\mathcal{P}}r_P$ and issues query ($P^*,I_\Delta$). 

After obtaining the records for query $(P,I_\Delta)$, we update the distribution $\pr(r|\theta_P, P)$, by updating the values of $\theta_P$ based on the records received. The distribution $\pr(r|\theta_P, P)$ before and after the update are  the {\em prior distribution} and {\em posterior distribution} respectively. We choose to use the Beta distribution \cite{johnson1995continuous} to model the prior distribution, which has been shown to be effective in a variety of settings \cite{russo2017tutorial}. The Beta distribution is characterized by two parameters $\alpha$ and $\beta$, and it is a particularly good fit for our problem as $\alpha$ and $\beta$ represent the pseudo counts of the number of correct and incorrect classifications we believe the CLF can make, providing our initial perspective of the reward function of $P$. 
Moreover, as shown in Section \ref{sec:bounding_error}, the reward of $(P,I_\Delta)$, i.e., the number of correctly classified records in $\mathcal{R}_P^{I_\Delta}$, follows a Binomial distribution.  It is known that the conjugate of Binomial distribution is the Beta distribution \cite{diaconis1979conjugate}, and the posterior distribution will still be a Beta distribution, making parameter update highly tractable. 
We next show how to compute a Beta posterior from a Beta prior and $\mathcal{R}_P^{I_\Delta}$.

Let Beta($\alpha$, $\beta$) be the Beta distribution with two parameters $\alpha$ and $\beta$. In our case, we initialize both $\alpha$ and $\beta$ to 1 for all predicates, essentially making the Beta distribution a uniform distribution, in line with the fact that the consumer has no knowledge regarding the rewards of predicates at the beginning. Suppose that the reward distribution for $P$ is Beta($\alpha_P,\beta_P$) before a round of interaction, and $\mathcal{R}_P^{I_\Delta}$ is received in this round. Let $N_P^{I_\Delta}$ be the number of records correctly labelled by CLF, i.e.,

\begin{equation}
    N_P^{I_\Delta}=\sum_{(x_i,y_i)\in\mathcal{R}_P^{I_\Delta}}\mathbb{I}[\mbox{CLF}((x_i,y_i))=0]
\end{equation}

We can show that $\alpha_P$ and $\beta_P$ can be updated as follows to obtain the posterior distribution conditional on $N_P^{I_\Delta}$ and $I_\Delta$:

\begin{equation}
\label{eq:post}
    (\alpha_P,\beta_P)\leftarrow (\alpha_P,\beta_P)+(N_P^{I_\Delta},I_\Delta-N_P^{I_\Delta})
\end{equation}

It follows immediately from Equation (\ref{eq:post}) that (1) the expectation of distribution Beta($\alpha_P,\beta_P$), computed as $\frac{\alpha_P}{\alpha_P+\beta_P}$, is proportional to the reward of $P$ (notice that $\frac{\alpha_P}{\alpha_P+\beta_P}$ is essentially the percentage of records satisfying $P$ that can be correctly labelled by CLF), so that the probabilities of selecting predicates with high observed rewards in future interactions are higher; and (2) after the update, $(\alpha_P+\beta_P+2)$ is equal to the number of records obtained using $P$ so far (as both $\alpha_P$ and $\beta_P$ are initialized to 1). With more records acquired using $P$, ($\alpha_P+\beta_P$) becomes larger and the distribution of Beta($\alpha_P,\beta_P$) becomes more concentrated, meaning that we are more confident regarding the expectation of $P$'s reward. Note that this is also the reason why we use the number of correctly labelled records as the reward rather than percentage thereof: even both queries ($P,I_\Delta=100$) and  ($P,I_\Delta=10$) return records of which 80\% can be correctly labelled, the former should give us more confidence regarding the distribution of $P$'s reward and thus ($\alpha_P,\beta_P$) should be greater in this case.

\subsection{Non-stationary Reward Distributions}
\label{sec:nonstationary}
For our discussion in Section~\ref{sec:spsts}, we have assumed that the reward distribution is  stationary regardless of the number of records acquired in previous rounds. However, one can observe that as  $\hat{\mathcal{R}}_P$ (the records the consumer has that satisfy predicate $P$) grows, the new information brought by each additional record from $\mathcal{D}_{pool} $ satisfying $P$ decreases, and consequently the reward of $P$ decreases. 
As such, not all past rewards observed should be treated equally. We should focus on the rewards observed from recent rounds, which better reflect the current reward distributions. Therefore, we modify the posterior computation in Equation (\ref{eq:post}) in a way that remembers only the rewards observed from the most recent $\tau$ rounds of interactions, as inspired by previous research on dealing with non-stationary reward distributions (e.g., \cite{russo2017tutorial,ikonomovska2015real}). 

More specifically, assume that the consumer has interacted with the provider using $P$ for $t$ rounds (including the current round), with rewards $N_P^{I_\Delta}[1],N_P^{I_\Delta}[2],\cdots,N_P^{I_\Delta}[t]$, then $\alpha_P$ and $\beta_P$ are updated as follows in two steps:

\begin{equation}
\label{eq:update_recent}
\begin{split}
    (\alpha_P,\beta_P)&\leftarrow(\alpha_P,\beta_P)+(N_P^{I_\Delta}[t],I_\Delta-N_P^{I_\Delta}[t]);\\
    (\alpha_P,\beta_P)&\leftarrow(\alpha_P,\beta_P)-(N_P^{I_\Delta}[t-\tau],I_\Delta-N_P^{I_\Delta}[t-\tau]), \text{ only if } t>\tau
\end{split}
\end{equation}

By remembering only the most recent rewards, the expectation of Beta($\alpha_P,\beta_P$) is closer to the current reward of predicate $P$. Besides, ''forgetting'' the previous rewards prevents Beta($\alpha_P,\beta_P$) from becoming too concentrated (recall that ($\alpha+\beta$) influences how concentrated the distribution is), and thus always allows a chance for more exploration, suitable for the setting with changing rewards.
\subsection{The SPS Algorithm}
The operation of SPS is summarized in Algorithm \ref{alg:sqs}. We initialize all reward distributions to Beta($1,1$) (Line 1). At each interaction, we randomly sample a value $r_P$ from distribution Beta$(\alpha_P,\beta_P)$ for each $P$ (Lines 3-4), and select the predicate $P^*$ with the highest $r_P$ and issue query $(P^*,I_\Delta)$ (Lines 5-6). After receiving a set of records, $\mathcal{R}_{P^*}^{I_\Delta}$, we update the values of $\alpha_P$ and $\beta_P$ accordingly (Lines 7-8), merge $\mathcal{R}_{P^*}^{I_\Delta}$ into acquired records and deduct $I_\Delta$ from the remaining budget (Line 9). The process terminates when the budget is exhausted.
\begin{algorithm}[htb]
\caption{Sequential Predicate Selection}
\label{alg:sqs}
\begin{algorithmic}[1]
\Require
budget $B$, all predicates $\mathcal{P}$
\Ensure
A set of records $\mathcal{R}$
\State \textbf{Initialization:} $\forall P\in\mathcal{P}$, $\alpha_P=1$, $\beta_P=1$; $\mathcal{R}=\emptyset$
\While {$B>0$}
\For {$P$ in $\mathcal{P}$}
\State sample $r_P$ from distribution Beta($\alpha_P,\beta_P$);
\EndFor
\State $P^*=\arg\max_{P\in\mathcal{P}}r_P$;
\State ask query ($P^*,I_\Delta$) and receive records $\mathcal{R}_{P^*}^{I_\Delta}$;
\State compute $N_{P^*}^{I_\Delta}=\sum_{(x_i,y_i)\in\mathcal{R}_{P^*}^{I_\Delta}}\mathbb{I}[\mbox{CLF}((x_i,y_i)=0)]$;
\State update $(\alpha_{P^*},\beta_{P^*})$ with Equation (\ref{eq:update_recent});
\State $\mathcal{R}=\mathcal{R}\cup\mathcal{R}_{P^*}^{I_\Delta}$; $B=B-I_\Delta$;
\EndWhile
\State \Return $\mathcal{R}$;
\end{algorithmic}
\end{algorithm}
\section{Experiments}
\label{sec:exp}
The techniques we propose are equally applicable to traditional Machine Learning  \cite{DBLP:conf/sigmod/ChenK019} and Deep Learning \cite{DBLP:conf/cvpr/YanAF20} models across a variety of domains. We choose models and datasets in the experiments to reflect the wide range of applications we envision. More specifically, we choose state-of-the-art deep models as well as classical ML models to experiment with. The datasets used in the experiments include image data, spatial data, and optical-radar data, and the tasks range from classification to regression.

\subsection{Settings}
\label{exp:settings}
{\bf Datasets.} We conduct experiments on four datasets. CIFAR10 and CIFAR100 \cite{krizhevsky2009learning} are image classification datasets widely used in the area of ML. {The Crop mapping dataset \cite{khosravi2019random} (Crop for short) contains temporal, spectral, textural, and polarimetric attributes for cropland classification. It has 175 real-valued features and one target (seven crop types).} 3D Road Network \cite{kaul2013building} (RoadNet for short) is a geographical dataset consisting of tuples of longitude, latitude, and altitude. {Following the instruction in \cite{kaul2013building}, we use longitude and latitude as the features and altitude as the predicted value.} The latter two datasets can also be found in the UCI data repository \cite{Dua:2019}. The characteristics of the four datasets are summarized in Table \ref{tab:dataset}.

CIFAR10, CIFAR100, and Crop are used for classification tasks, while RoadNet is used for a regression task. To generate $\mathcal{D}_{test}$, we directly use the test sets provided by CIFAR10 and CIFAR100, and randomly select $20\%$ records from Crop and RoadNet.
\begin{table}[h]
\caption{Dataset Characteristics}
\begin{tabular}{|c|c|c|c|}
\hline
\textit{dataset}  & \textit{$\#$ records} & \textit{$\#$ classes} & $\#$ \textit{dimensions} \\ 
\hhline{|=|=|=|=|}
CIFAR10  & 60,000       & 10           & 1,024            \\ \hline
CIFAR100 & 60,000       & 100          & 1,024           \\ \hline
Crop     & 325,834       & 7            & 175             \\ \hline
RoadNet  & 434,874      & N/A          & 2              \\ \hline
\end{tabular}
\label{tab:dataset}
\end{table}

{\bf Models.} For classification on CIFAR10 and CIFAR100, we adopt VGG8B with predsim loss function \cite{DBLP:conf/icml/NoklandE19}, one of the most recent and state-of-the-art deep learning structures. For classification on Crop, we use Decision Tree. For RoadNet, we use $k$NN Regressor \cite{altman1992introduction}. We also utilized other applicable models for our evaluation (e.g., AlexNet \cite{DBLP:conf/nips/KrizhevskySH12}, EfficientNet \cite{DBLP:conf/icml/TanL19}) and observed similar trends. We use the code of VGG8B provided by the authors, and the Decision Tree and $k$NN Regressor implementation in scikit-learn. Default settings are adopted for all models. 

{\bf Construction of $\mathcal{P}$.} For CIFAR10, CIFAR100, and Crop, we build predicates based on class labels, resulting in 10 predicates for CIFAR10, 100 predicates for CIFAR100, and 7 predicates for Crop. {For RoadNet, by default we discretize the data space by partitioning the range of each feature into four equal-width sub-ranges, resulting in $4^2=16$ cells; a predicate selects records falling into a specific cell. We study the influence of $|\mathcal{P}|$ in Section \ref{exp:partition}}.

{\bf Construction of $\mathcal{D}_{init}$.} We construct $\mathcal{D}_{init}$ using $20\%$ records in the corresponding dataset $\mathcal{D}$ for CIFAR10 and CIFAR100, and $1\%$ records for Crop and RoadNet. { We select records from $\mathcal{R}_P$ for each $P\in\mathcal{P}$ to construct $\mathcal{D}_{init}$ following a power-law distribution. More specifically, with a random order of predicates in $\mathcal{P}$, let $P_i$ be the $i$-th predicate ($i\in[1,\vert\mathcal{P}\vert]$), the number of records selected from $\mathcal{R}_{P_i}$ is proportional to $i$.}

{\bf Selection of CLF.} {We use the $k$NN classifier with $k=1$ as the CLF in our experiments. We study the impact of varying $k$ values as well as utilizing diverse classifiers in Section \ref{sec:exp-clf}.} For Crop and RoadNet, we directly use the raw attributes as the input of CLF. In accordance to previous work on image-based $k$NN search \cite{wang2014learning}, for CIFAR10 and CIFAR100, we first use HOG \cite{dalal2005histograms} (Histogram of Oriented Gradients), a widely-adopted image feature extractor, to transform an image into a feature vector, and use the transformed feature vectors as the input of CLF.

{\bf Evaluation.} Let $\mathcal{D}_{otd}$ be the records acquired during the acquisition process. We train the model on $\mathcal{D}_{init}\cup\mathcal{D}_{otd}$ and evaluate on $\mathcal{D}_{test}$. For the classification tasks on CIFAR10, CIFAR100, and Crop, accuracy is computed as follows:

\begin{equation}
    accuracy=\frac{\sum_{(x,y)\in\mathcal{D}_{test}}\mathbb{I}[\mathcal{M}(x)=y]}{|\mathcal{D}_{test}|}
\end{equation}

where $\mathcal{M}(x)$ denotes the model output on $x$.

For the regression task on RoadNet, we use $R^2$ score to evaluate the performance, computed as follows:

\begin{equation}
    R^2=1-\frac{\sum_{(x,y)\in\mathcal{D}_{test}}(y-\mathcal{M}(x))^2}{\sum_{(x,y)\in\mathcal{D}_{test}}(y-\bar{y})^2}
\end{equation}

where $\bar{y}=\frac{\sum_{(x,y)\in\mathcal{D}_{test}}y}{|\mathcal{D}_{test}|}$.

The acquisition process and model training are repeated ten times and the average accuracy/$R^2$ score is reported.

\subsection{The Effect of $l$ on EA}
\label{exp:l}
As described in Section \ref{sec:ea}, EA requires acquiring $l\%$ random records for each predicate from the provider for initialization. Here, we experimentally evaluate the effect of $l$ on the performance of EA. To have a common basis for the evaluation of the trends we fix the significance level ($\delta$) at 0.001 throughout the experiment. The results are provided in Figure \ref{fig:ea-init}.
\begin{figure}[h]
    \centering
    \includegraphics[width=3.2in]{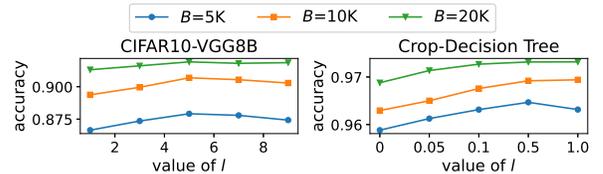}
    \caption{Effect of $l$ on EA}
    \label{fig:ea-init}
\end{figure}

As demonstrated in Figure \ref{fig:ea-init}, the performance of EA with a small budget is more sensitive to the choice of $l$, while $l$ has a less significant impact on EA's performance when a large budget is used, except for cases where $l$ is very small. The trade-off involved in selecting $l$ can be summarized as follows. Using an overly-large $l$ would result in too much budget consumption for the initialization, and consequently reduce the budget available for the Allocation Stage of EA, especially when the total budget $B$ is small. On the other hand, using too small an $l$ may cause the Estimation Stage to perform badly, leading to low-quality predicate utility estimates, as discussed in Section \ref{sec:ea_alg}. In the following experiments, we set $l=5$ for CIFAR10 and CIFAR100, and $l=0.5$ for RoadNet and Crop. Since the budgets we use are fairly large for the respective dataset, the performance of EA is less dependent on the choice of $l$.

{{\bf Takeaways.} (1) An overly small or large $l$ may harm the performance of EA. The value $l=5$ for image data, and $l=0.5$ for spatial data worked best during our experiments; (2) The performance of EA becomes less dependent on the choice of $l$ as budget increases.}

{
\subsection{Estimation Accuracy of EA}
In the Estimation Stage of EA, we assess the utility of predicate $P$, $U_P$, based on the records acquired with $P$, $\hat{\mathcal{R}}_P$. Let $\hat{U}_P$ be the estimated value of $U_P$. In this section, we examine the estimation accuracy. More specifically, we vary the number of records used for utility estimation, {normalized by $|{\mathcal{R}}_{P}|$, i.e., ${|\hat{\mathcal{R}}_P|}/{|{\mathcal{R}}_{P}|}$}, to study its effect on the absolute error of the estimation, i.e., $|U_P-\hat{U}_P|$. The results are presented in Figure \ref{fig:ea-estimation-accuracy}.

\begin{figure}[h]
    \centering
    \includegraphics[width=2in]{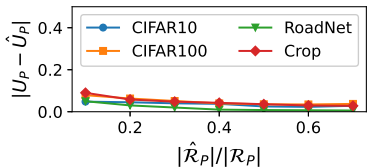}
    \caption{Estimation Accuracy}
    \label{fig:ea-estimation-accuracy}
\end{figure}

As can be observed from Figure \ref{fig:ea-estimation-accuracy}, the absolute error of the estimation is consistently low (below 0.1), and increasing the number of records used for estimation further reduces the absolute error. The reason is that, according to Theorem \ref{theorem:normal}, $U_P-\hat{U}_P$ follows a normal distribution, and increasing $|\hat{\mathcal{R}}_P|$ reduces the standard deviation of the distribution; consequently, the value of $\hat{U}_P$ {converges to} the true utility, $U_P$.
}

\subsection{Comparison of Allocation Strategies in EA}
\label{exp:allocation}
Two strategies have been proposed in Section \ref{sec:allocation}, namely, Linear Allocation and Square-root Allocation, which allocate the budget proportional to the estimated utility of each predicate or its square-root respectively. In this section we evaluate the impact of the allocation strategy on the performance of EA, and present the results in Figure \ref{fig:allocation}. { We also conduct 
one-tailed t-tests to determine whether the average accuracy improvement of one allocation strategy is statistically higher than the other. The cases for $B=5,000$ and $B=20,000$ are provided in Table \ref{table:sig-test-alloc}, where $\mu_L$ ($\mu_S$) denotes the average accuracy improvement of Linear (Square-root) Allocation.}

\noindent\begin{minipage}{.26\textwidth}
\centering
\includegraphics[width=\textwidth]{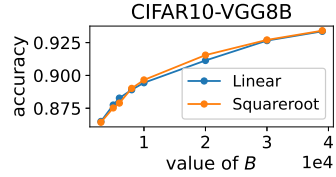}
\captionof{figure}{Allocation Strategies}
\label{fig:allocation}         
\end{minipage}
\noindent\begin{minipage}{.22\textwidth}
	\captionof{table}{Significance Test}
	\label{table:sig-test-alloc}
	\centering
	\begin{tabular}{ccc}
		\toprule
		$B$      &$H_A$    & p-value \\
		\midrule
		5,000    & $\mu_L>\mu_S$ & 4e-3   \\
		20,000  & $\mu_S>\mu_L$ & 1e-7   \\
		\bottomrule
	\end{tabular}
\end{minipage}

{As is clear from Figure \ref{fig:allocation} and Table \ref{table:sig-test-alloc}, different allocation strategies achieve similar performance overall, and with significance level $\alpha<0.01$ we say Linear Allocation yields higher accuracy when the budget is relatively small (e.g., $B=5,000$) and Square-root Allocation takes the lead when the budget is large (e.g., $B=20,000$).} This observation is the result of two competing underlying factors: on one hand, we should exploit the knowledge on the predicate utility gained from the Estimation Stage and therefore we should bias the allocation towards predicates with higher estimated utilities as much as possible (hence the superiority of Linear Allocation over Square-root Allocation for small budgets); on the other hand, as discussed in Section \ref{sec:sps}, the utility of a predicate $P$ decreases as more records are acquired with $P$, i.e., the marginal benefit of obtaining records using predicates with higher estimated utilities becomes lower as the budget grows (hence Square-root Allocation outperforms Linear Allocation for large budgets). In other experiments we adopt Linear Allocation.

{{\bf Takeaways.} (1) Linear Allocation is better suited for data acquisition with budget $B\le 0.2|\mathcal{D}|$ in our experiments; (2) Square-root Allocation  is better suited for data acquisition with budget $B> 0.2|\mathcal{D}|$ during our evaluation.}

\subsection{The Effect of Batch Size on SPS}
\label{exp:batch-size}
With SPS, records are acquired in small batches of size $I_\Delta$. We evaluate the effect of $I_\Delta$ on the performance of SPS presenting our results in Figure \ref{fig:batch_size}.
\begin{figure}[h]
    \centering
    \includegraphics[width=3.2in]{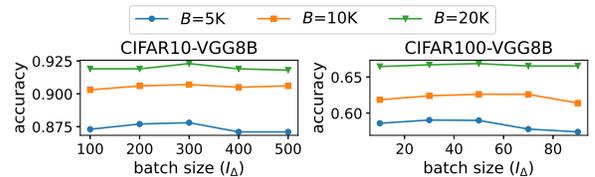}
    \caption{Effect of Batch Size on SPS}
    \label{fig:batch_size}
\end{figure}

As can be observed from Figure \ref{fig:batch_size}, data acquisition with a small budget is more sensitive to the value of $I_\Delta$. The trade-off in selecting $I_\Delta$ is as follows. Let $(P,I_\Delta)$ be the query. With a small $I_\Delta$, the records returned by the provider may not be representative of $\mathcal{R}_P$, and the query reward thus computed is inaccurate. Consequently SPS cannot identify those predicates with high rewards. On the other hand, if $I_\Delta$ is too large, then each interaction consumes too much budget, limiting the exploitation of predicates with higher rewards.  However, as the budget increases, SPS becomes progressively less sensitive to the batch size, because (1) the variations in the records obtained in each interaction have minimal impact on the overall estimation accuracy given a large number of interactions; (2) although exploration with large $I_\Delta$ consumes more budget, it also provides more information regarding the reward distribution (see Equation (\ref{eq:post})), benefiting future predicate selection. In the following experiments, we select $I_\Delta=300$ for CIFAR10, and $I_\Delta=30$ for CIFAR100, Crop and RoadNet. Since the budgets we use are fairly large for the respective datasets, we expect less dependency on the choices of $I_\Delta$.

{{\bf Takeaways.} (1) The accuracy of SPS is relatively stable {over a wide range of batch size values, only slightly decreasing for an overly small or large batch size, depending on the ML/DL model and data}; (2) The accuracy of SPS becomes less dependent on the batch size as budget increases.}
\subsection{The Effect of $\tau$ on SPS}
As discussed in Section \ref{sec:nonstationary}, to deal with the non-stationary reward, we only use the records acquired by the most recent $\tau$ queries with $P$ to update the posterior distribution of $P$'s reward. We evaluate the effect of $\tau$ on the performance of SPS, and report the results on CIFAR10 in Figure \ref{fig:tau}; similar trends can be observed on other datasets. 
\begin{figure}[h]
    \centering
    \includegraphics[width=2.5in]{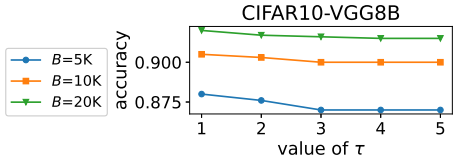}
    \caption{Effect of $\tau$ on SPS}
    \label{fig:tau}
\end{figure}

As can be observed from Figure \ref{fig:tau}, the performance of SPS is relatively stable across different values of $\tau$, with small $\tau$ yielding slightly higher accuracy. { The reason is that, as indicated by Equation (\ref{eq:update_recent}), with larger $\tau$, $\alpha_P$ and $\beta_P$ contain more dated reward observations and thus diverge from the current reward distribution; this in turn may cause the acquisition of less useful records (in terms of accuracy improvement). Having less useful records clearly has a stronger influence on model accuracy when the total number of training records is smaller (corresponding to a smaller budget). However, with the SPS strategy, as long as the predicate with the actual maximal expected reward has a higher probability to be selected than the other predicates, the cumulative reward is likely to be maximized. Therefore, SPS is tolerant to inaccurate reward distribution estimations and robust to the value of $\tau$.} We set $\tau$ to 1 in the following experiments.

{ {\bf Takeaways.} (1) Setting $\tau=1$ always leads to higher accuracy during our evaluation; (2) SPS is relatively robust with respect to $\tau$.}
{
\subsection{The Effect of $|\mathcal{P}|$}
\label{exp:partition}

In this section we study the influence of the number of predicates, i.e., $|\mathcal{P}|$, on the accuracy of the methods. More specifically, we change $|\mathcal{P}|$ by changing either: (1) the number of labels to construct $\mathcal{P}$ for a classification dataset, or (2) the discretization granularity for a regression dataset. For case (1), we select CIFAR100 and use the $m$ labels with which the model has the lowest accuracy (cross validated on $\mathcal{D}_{init}$) to construct $\mathcal{P}$. For case (2), we use RoadNet and partition the range of each of the two features into $n$ equal-width sub-ranges, resulting in $n^2$ cells, which are used as $\mathcal{P}$. The results are presented in Figure \ref{fig:partition}.

\begin{figure}[h]
    \centering
    \includegraphics[width=3.2in]{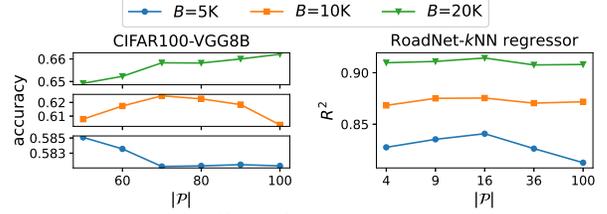}
    \caption{Effect of Space Discretization}
    \label{fig:partition}
\end{figure}

As is clear from the results on CIFAR100, for case (1), the trend of accuracy with respect to $|\mathcal{P}|$ depends on the budget. The $|\mathcal{P}|$ yielding the maximal accuracy is 50 for $B$=5K, 70 for $B$=10K, and 100 for $B$=20K. {The trend confirms our intuition that with a small budget, limiting the data acquisition to the predicates where the model is more error-prone reduces extra exploration cost and consequently increases the accuracy; however, this may result in over-exploitation of these predicates when the budget is large due to decreasing utility (as discussed in Section \ref{sec:nonstationary}), leading to a slower increase in accuracy compared to larger $\mathcal{P}$.} As can be observed from the results on RoadNet, for case (2), the $R^2$ score is relatively stable with respect to $|\mathcal{P}|$, with a slight increase when $|\mathcal{P}|$ is between 9 and 36. 
The reason is that, an overly coarse partitioning granularity (small $|\mathcal{P}|$) prevents the effective identification of the area in the data space where the model has a low $R^2$ score, while an overly fine partitioning granularity (large $|\mathcal{P}|$) increases the exploration cost as there are more predicates whose utilities need to be estimated. 

{\bf Takeaways.} During our experiments, (1) Descritizing the data space such that each cell occupies $3\%-10\%$ of the size of the entire data space gives higher $R^2$ score; (2) For classification tasks with budget $B\le 0.2|\mathcal{D}|$, using no more than $70\%$ of all labels to construct $\mathcal{P}$ gives higher accuracy; with $B> 0.2|\mathcal{D}|$, using at least $70\%$ of all labels to construct $\mathcal{P}$ leads to higher accuracy;  (3) The performance of the proposed methods is generally stable in terms of the number of predicates.}




\subsection{EA vs. SPS}
We now experimentally compare the two methods proposed in the paper, EA and SPS, and showcase the relative trends. We report the results in Figure \ref{fig:ea-sps}. { We also conduct one-sided t-tests to determine whether the average accuracy improvement of one method is statistically higher than the other. The cases for $B=3,000$ and $B=20,000$ are provided in Table \ref{table:sig-test-ea-sps}, where $\mu_{P}$ ($\mu_{E}$) denotes the average accuracy improvement of SPS (EA).}

\noindent\begin{minipage}{.26\textwidth}
\centering
\includegraphics[width=\textwidth]{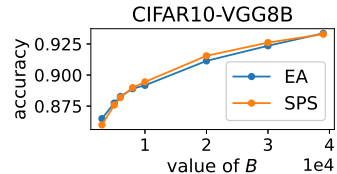}
\captionof{figure}{EA vs. SPS}
\label{fig:ea-sps}        
\end{minipage}
\noindent\begin{minipage}{.22\textwidth}
	\captionof{table}{Significance Test}
	\label{table:sig-test-ea-sps}
	\centering
	\begin{tabular}{ccc}
		\toprule
		$B$      &$H_A$    & p-value \\
		\midrule
		3,000    & $\mu_{E}>\mu_{P}$ & 1e-4   \\
		20,000  & $\mu_{P}>\mu_{E}$& 1e-9   \\
		\bottomrule
	\end{tabular}
\end{minipage}\hfill

{As can be observed from Figure \ref{fig:ea-sps} and Table \ref{table:sig-test-ea-sps}, EA and SPS provide similar performance, and with significance level $\alpha<0.001$ we say EA provides a higher accuracy gain with a small budget (e.g., $B=3,000$) and SPS provides a higher accuracy improvement with a large budget (e.g., $B=20,000$).} The reason is that EA is a budget-aware method: with a small budget it tends to allocate less budget for utility estimation, aided by the heuristic reward, so that more budget can be allocated to predicates with high estimated utilities. The mechanism of SPS, on the other hand, is budget-agnostic, and the exploration of all predicate rewards at the start consumes budget and limits the chances to exploit predicates with high utilities, especially when the budget is small. As the budget increases, SPS starts to outperform EA since it acquires records in small batches and can flexibly adjust the acquisition strategy in face of utility changes; EA conducts one-time allocation without considering future utility changes and the records thus obtained may not be impactful to improve model accuracy.

{{\bf Takeaways.} In our evaluation, (1) EA is better suited for data acquisition with budget $B\le 0.2|\mathcal{D}|$; (2) SPS is better suited for data acquisition with budget $B> 0.2|\mathcal{D}|$.}

{
\subsection{The Effect of CLF}
\label{sec:exp-clf}
As discussed in Section \ref{sec:preliminaries}, the utility of a predicate is essentially the accuracy of a classifier (CLF) in differentiating $\mathcal{R}_{P}$ and $\mathcal{R}_{\mathcal{M}:P}$. Here we experimentally study the influence of CLF on the accuracy improvement. More specifically, since the utility computation is carried out frequently, we consider lightweight models including the $k$NN classifier ($k\in\{1,3,5\}$), the decision tree classifier, and the perceptron classifier. We report results on CIFAR10 in Figure \ref{fig:clf}; similar trends can be observed on other datasets. 
\begin{figure}[h]
    \centering
    \includegraphics[width=3.2in]{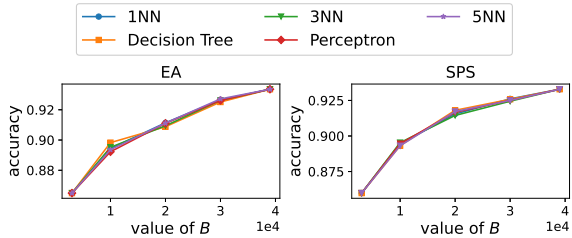}
    \caption{Effect of CLF}
    \label{fig:clf}
\end{figure}

As can be observed from Figure \ref{fig:clf}, the performance of the methods is not sensitive to the particular model used as the CLF and its hyper-parameters. In other experiments, we use the $k$NN classifier with $k=1$.
}

\subsection{Comparison with Baseline Methods for Data Acquisition}

The closest piece of work that can be adapted to our setting for comparison purposes is the work on active class selection (ACS) \cite{Lomasky2007ECML}. Although the problem solved therein is different, we can adapt the methods used for our setting. Therefore, we use ACS (adapted to our setting) as the baseline and present experimental results comparing it with our proposals. We note, however, that this is not a fair comparison, because ACS requires re-training the model after each interaction which can be computationally prohibitive for complex models involving large datasets, whereas ours does not. 

Specifically, ACS acquires $b$ new data records in each round (the same batch size as used in Section \ref{exp:batch-size}), which is allocated to each class uniformly (ACS-Uniform), or in proportion with the accuracy improvement for each class (ACS-AI), or the number of records in each class whose label has changed (ACS-RD) during the last round. Note that ACS-AI and ACS-RD require model re-training after new records are obtained and thus are too expensive to be applied to CIFAR10 and CIFAR100. As such, we apply ACS-Uniform to CIFAR10, and ACS-AI and ACS-RD to Crop, with results provided in Figure \ref{fig:baselines}. The observations on other datasets are similar and are thus omitted.


\begin{figure}[h]
    \centering
    \includegraphics[width=3.2in]{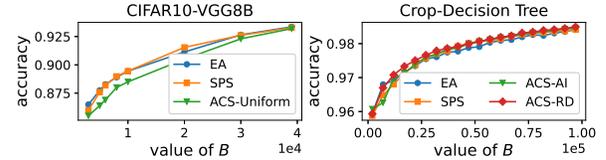}
    \caption{Comparison with Baselines}
    \label{fig:baselines}
\end{figure}

The results in Figure \ref{fig:baselines} (a) indicate that EA and SPS consistently outperform ACS-Uniform (except for the case of $B=4,000$ when all samples are acquired) and achieve similar accuracy to baselines {\em that require the model to be re-trained after each interaction}, by effectively acquiring records with higher novelty that are more likely to boost accuracy.

{
\subsection{Comparison with Utility Measures Based on Re-training/refinement}
\label{exp:measure}
One of the main advantages of the utility measure we propose, novelty, is that it does not require the computationally expensive step of model re-training. In this section, we compare novelty with re-training-based utility measures in terms of model accuracy improvement and data acquisition cost. More specifically, we compare with a re-training-based measure where $\mathcal{M}$ is re-trained on newly-acquired records, say $\mathcal{R}_P^I$,  and the improvement in accuracy after re-training is used as the utility of $P$. While all lightweight models are re-trained from scratch, we adopt the state-of-the-art incremental learning method, UCB \cite{DBLP:conf/iclr/EbrahimiEDR20}, to refine deep models instead of conducting complete re-training to keep training overhead manageable. We report the results of using SPS together with both utility measures on CIFAR10 and Crop in Figure \ref{fig:retraining}; observations are similar on other datasets and consistent in the case of using EA.

\begin{figure}[h]
    \centering
    \includegraphics[width=3.2in]{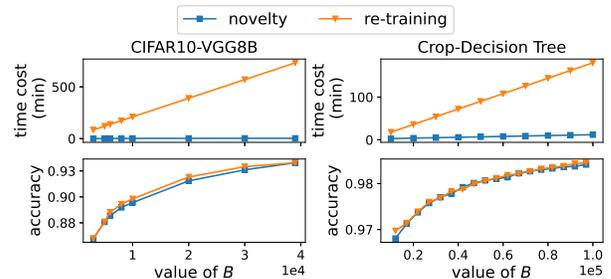}
    \caption{Comparison with Re-training-based Measure}
    \label{fig:retraining}
\end{figure}

The results in Figure \ref{fig:retraining} indicate that novelty achieves similar accuracy to the re-training-based utility measure by effectively acquiring records with higher novelty that are more likely to boost accuracy, while requiring orders of magnitude less time. Although the model can be refined incrementally with UCB, repetitive model refinement, especially when the budget is large, still results in high execution overhead. Although lightweight models such as decision trees can be constructed rapidly, repetitive construction imposes large computational overheads during the acquisition process. Novelty, on the contrary, conducts data acquisition by only looking at the data, and thus is highly efficient for practical deployments.
}
\section{Conclusions and Future Work}

In this paper, we have considered the problem of acquiring data in order to improve the accuracy of ML models, and laid out the framework of interaction between a provider and a consumer in the context of data markets. We have proposed two algorithmic solutions that the consumer with a limited budget could use to obtain data from the provider, striking a balance between exploration (gaining more knowledge on the data the provider possesses) and exploitation (utilizing that knowledge for allocating the limited budget for data acquisition). The first solution, EA, has two distinctive stages, Estimation Stage and Allocation Stage, focusing on exploration (obtaining accurate estimates on the predicate utilities) and exploitation (allocating the budget according to the estimates) respectively. The second solution, SPS, blends exploration and exploitation in each round of interaction, and adaptively allocates budget for the next round, investing resources into more promising areas of the data space to improve model accuracy. Results from our experimental studies have confirmed the effectiveness of our proposals, and illustrated the trade-offs and relative strengths of each proposed solution. 

Our work represents the first step in dealing with data acquisition in a market setting, and research opportunities in this new area abound. For example, one could consider the problem of acquiring data from multiple providers with varying data coverage and data quality, or investigate new mechanisms for data acquisition where there is a third party (e.g., data broker) involved. 
\appendix
\section{Notation table}
\label{sec:notation}
The notations and abbreviations are provided in Table \ref{tab:notation}.
\begin{table*}[h]

\caption{Notations and Abbreviations} \label{tab:notation} 
\begin{center}
\begin{tabular}{|c| c|}
\hline
{\bf Notation/Abbreviation} & {\bf Definition}\\

\hline
$\Gamma=\{\mathcal{X},\mathcal{Y}\}$ & the data domain, where $\mathcal{X}$ denotes the feature space and $\mathcal{Y}$ denotes the label space\\
\hline
$\mathcal{M}$ & the Machine Learning model\\
\hline
$\mathcal{D}_{pool}$ & the data the provider maintains\\
\hline
$\mathcal{D}_{init}$ & the data the consumer initially possesses\\
\hline
$\mathcal{D}_{test}$ & the test data on which the performance of $\mathcal{M}$ is evaluated\\
\hline
$P$ ($\mathcal{P}$) & an arbitrary predicate (all predicates)\\
\hline
$\mathcal{R}_P$ & all records in $\mathcal{D}_{pool}$ satisfying $P$\\
\hline
$Q=(P,I)$ & a query\\
\hline
$\mathcal{R}_Q$ & records returned to the consumer in response to query $Q$\\
\hline
$U_P$ & the utility (novelty) of predicate $P$\\
\hline
CLF & the classifier used in computing $U_P$\\
\hline
EA & the Estimation-and-Allocation strategy\\
\hline
$l$ & the percentage of records acquired from $\mathcal{R}_P$ during the initialization stage of EA\\
\hline
SPS & the Sequential Predicate Selection strategy\\
\hline
$\tau$ & the number of rounds to remember in SPS\\
\hline
\end{tabular}
\end{center}
\end{table*}
\bibliographystyle{ACM-Reference-Format}
\bibliography{reference}
\end{document}